\documentclass[12pt,reqno]{article}

\usepackage[usenames]{color}
\usepackage{amssymb}
\usepackage{amsmath}
\usepackage{amsthm}
\usepackage{amsfonts}
\usepackage{amscd}
\usepackage{graphicx}
\usepackage{xcolor}

\usepackage[colorlinks=true,
linkcolor=webgreen,
filecolor=webbrown,
citecolor=webgreen]{hyperref}

\definecolor{webgreen}{rgb}{0,.5,0}
\definecolor{webbrown}{rgb}{.6,0,0}

\usepackage{color}
\usepackage{fullpage}
\usepackage{float}

\usepackage{graphics}
\usepackage{latexsym}
\usepackage{epsf}
\usepackage{breakurl}
\usepackage{fullpage}

\newcommand{\seqnum}[1]{\href{https://oeis.org/#1}{\rm \underline{#1}}}

\DeclareMathOperator{\feq}{feq}
\DeclareMathOperator{\feqc}{feqc}
\DeclareMathOperator{\either}{either}
\DeclareMathOperator{\consec}{consec}
\DeclareMathOperator{\ab}{ab}
\DeclareMathOperator{\first}{first}
\DeclareMathOperator{\afirst}{afirst}
\DeclareMathOperator{\abpat}{abpat}
\DeclareMathOperator{\bapat}{bapat}

\newenvironment{smallarray}[1]
{\null\,\vcenter\bgroup\scriptsize
\arraycolsep=.13885em
\hbox\bgroup$\array{@{}#1@{}}}
{\endarray$\egroup\egroup\,\null}

\begin{document}

\theoremstyle{plain}
\newtheorem{theorem}{Theorem}
\newtheorem{corollary}[theorem]{Corollary}
\newtheorem{lemma}[theorem]{Lemma}
\newtheorem{proposition}[theorem]{Proposition}

\theoremstyle{definition}
\newtheorem{definition}[theorem]{Definition}
\newtheorem{example}[theorem]{Example}
\newtheorem{conjecture}[theorem]{Conjecture}

\theoremstyle{remark}
\newtheorem{remark}[theorem]{Remark}

\title{Intertwining of Complementary Thue-Morse Factors}

\author{Jeffrey Shallit\\
School of Computer Science\\
University of Waterloo\\
Waterloo, ON  N2L 3G1\\
Canada\\
\href{mailto:shallit@uwaterloo.ca}{\tt shallit@uwaterloo.ca}}

\maketitle

\begin{abstract}
We consider the positions of
occurrences of a factor $x$ and its binary
complement $\overline{x}$ in the Thue-Morse word
${\bf t} = {\tt 01101001} \cdots$, and show
that these occurrences are ``intertwined'' in essentially
two different ways.   Our proof method consists of stating
the needed properties as a first-order logic formula $\varphi$,
and then using a theorem-prover to prove $\varphi$.
\end{abstract}

\section{Introduction}

The Thue-Morse sequence ${\bf t} = t_0 t_1 t_2 \cdots = {\tt 01101001} \cdots$
is a famous binary sequence with many interesting properties
\cite{Allouche&Shallit:1999}.  
In this short note we prove yet
another in a long list of such properties, this time concerning
complementary factors.

A {\it factor\/} of an infinite word $\bf w$ is a contiguous block sitting
inside $\bf w$.  In this paper we will only be concerned with 
factors of finite length.   Define $\overline{\tt 0} = {\tt 1}$ and
$\overline{\tt 1} = {\tt 0}$, and extend this notion to words in the
obvious way, so that if $w = a_1 a_2 \cdots a_n$, then
$\overline{w} = \overline{a_1} \, \overline{a_2} \, \cdots \, 
\overline{a_n}$.   We say two binary words $x, y$ are {\it complementary}
if $x = \overline{y}$.   Thus, for example, {\tt 0110} and {\tt 1001} are
complementary.

It is well known that the Thue-Morse word $\bf t$ is {\it recurrent}, that is,
every factor that occurs, occurs infinitely often (first observed by
Morse \cite{Morse:1921}).  Further, it is
complement-invariant:  if a factor $x$ occurs in $\bf t$, then so does its
binary complement $\overline{x}$.   This suggests looking
at the positions of the occurrences of $x$ and $\overline{x}$ in $\bf t$.

For example,
let us do this for the complementary factors ${\tt 00}$
and ${\tt 11}$, marking the occurrences of ${\tt 00}$ in 
\textcolor{red}{red}
and ${\tt 11}$ in \textcolor{blue}{blue}:
$$ {\tt 0} \textcolor{blue}{\tt 11} {\tt 01} \textcolor{red}{\tt 00} \textcolor{blue}{\tt 11} 
\textcolor{red}{\tt 00} {\tt 10} \textcolor{blue}{\tt 11} {\tt 01} \textcolor{red}{\tt 00} {\tt 1} \cdots$$
Seeing this, it is natural to conjecture that occurrences of {\tt 11} and
{\tt 00} strictly alternate in $\bf t$, a conjecture that is not hard
to prove and appears in \cite{Bernhardt:2009}.
 Spiegelhofer
\cite[Lemma 2.10]{Spiegelhofer:2021} handled the case of {\tt 01} and {\tt 10}.

However, strict alternation, as in this example, is not the only possibility
for other factors.
If we consider the complementary factors {\tt 01101} and {\tt 10010} instead,
then the occurrences behave differently:
$$ \textcolor{red}{\tt 01101} {\tt 001100} \textcolor{blue}{\tt 10110} {\tt 100}
\textcolor{blue}{\tt 10110} \textcolor{red}{\tt 01101} {\tt 001} \cdots .$$
If we write $\tt A$ for an occurrence of {\tt 01101} and $\tt B$ for
an occurrence of its complementary factor, then experiments quickly
lead to the conjecture that these factors occur in the repeating
pattern ${\tt (ABBA)}^\omega = {\tt ABBA ABBA ABBA } \cdots $.

In this paper we prove that the two patterns ${\tt (AB)}^\omega$ and
${\tt (ABBA)}^\omega$ are essentially the only nontrivial possibilities
for intertwining of complementary factors.

For a deep study of the gaps between successive
occurrences of factors in $\bf t$, see the recent
paper of Spiegelhofer \cite{Spiegelhofer:2021}.

\section{The main theorem}

Let $x$ be a finite, nonempty factor of the Thue-Morse word $\bf t$.
Consider all occurrences of $x$ and $\overline{x}$ in
$\bf t$  and identify their starting positions, writing
$\tt A$ for an occurrence of $x$ and $\tt B$ for an occurrence
of $\overline{x}$.   (An occurrence of $x$ may overlap that of $\overline{x}$.)
Call the resulting infinite sequence of {\tt A}'s and {\tt B}'s the
{\it intertwining sequence} of $x$, and write it as $I(x)$.

The following is our main result.
\begin{theorem}
The only possibilities for $I(x)$ are as follows:
\begin{enumerate}
\item ${\tt ABBABAABBAABABBA}\cdots $, which is the Thue-Morse word itself
under the coding ${\tt 0} \rightarrow {\tt A}$,
${\tt 1} \rightarrow {\tt B}$;
\item ${\tt BAABABBAABBABAAB}\cdots$, which is the Thue-Morse word itself
under the coding ${\tt 0} \rightarrow {\tt B}$,
${\tt 1} \rightarrow {\tt A}$;
\item ${\tt (AB)}^\omega$;
\item ${\tt (BA)}^\omega$;
\item ${\tt (ABBA)}^\omega$;
\item ${\tt (BAAB)}^\omega$.
\end{enumerate}
Furthermore, possibility 1 only occurs if $x = {\tt 0}$ and
possibility 2 only occurs if $x = {\tt 1}$.
\label{main}
\end{theorem}

\begin{proof}
It is trivial to see the claim for $x = {\tt 0}$ and $x = {\tt 1}$.
So in what follows, we assume $|x| \geq 2$.

The idea of our proof is to write first-order logic formulas
for assertions that imply our desired results, and then use
the theorem-prover {\tt Walnut} to prove the results.
This is a strategy that has been used many times now
(see, e.g., \cite{Shallit:2022}).
For more about {\tt Walnut}, see \cite{Mousavi:2016}.

We describe the formulas in detail for the cases ${\tt (AB)}^\omega$
and ${\tt (ABBA)}^\omega$, leaving the other cases to the reader.

To assert that the pattern ${\tt (AB)}^\omega$ describes the
occurrences of $x$ and $\overline{x}$ in $\bf t$, we create
first-order logic formulas asserting each of the following:
\begin{itemize}
\item[(a)] 
one of the two
words $x$ and $\overline{x}$ occurs at positions $j, k$ for $j<k$, and
furthermore that neither of the two words occurs
at any position between $j$ and $k$.   This ensures
that $j$ and $k$ mark the starting position of two {\it consecutive\/}
factors chosen from $\{ x, \overline{x} \}$.

\item[(b)] if $j,k$ are two
positions as in (a), then one must be the position of
$x$, while the other is the position of $\overline{x}$.
This forces the consecutive occurrences of the factors to
alternate, and hence form either the pattern ${\tt (AB)}^\omega$
or ${\tt (BA)}^\omega$.

\item[(c)] the first occurrence
of either $x$ or $\overline{x}$ in $\bf t$ is actually an
occurrence of $x$.  This, together with (b),
forces the pattern to be of the form ${\tt (AB)}^\omega$.
\end{itemize}

We specify the word $x$ by giving one of its occurrences, that is,
two integers $i, n$ such that
$x= {\bf t}[i..i+n-1]$.  

Here is the meaning of each logical formula we now define.  

\begin{itemize}
\item $\feq(i,j,n)$ asserts that ${\bf t}[i..i+n-1]={\bf t}[j..j+n-1]$;
\item $\feqc(i,j,n)$ asserts that ${\bf t}[i..i+n-1] = \overline{{\bf t}[j..j+n-1]}$;
\item $\either(i,j,n)$ asserts that either ${\bf t}[i..i+n-1]={\bf t}[j..j+n-1]$
or ${\bf t}[i..i+n-1] = \overline{{\bf t}[j..j+n-1]}$;
\item $\consec(i,j,k,n)$ asserts that $j<k$ and
${\bf t}[j..j+n-1] \in \{x, \overline{x} \}$ and
${\bf t}[k..k+n-1] \in \{x, \overline{x} \}$, where
$x = {\bf t}[i..i+n-1]$,
but no factor starting in between these two equals
either $x$ or $\overline{x}$.
\item $\ab(i,j,k,n)$ asserts 
${\bf t}[j..j+n-1] = x$ and
${\bf t}[k..k+n-1] = \overline{x}$,
for $x = {\bf t}[i..i+n-1]$.
\item $\first(i,j,n)$ asserts that ${\bf t}[j..j+n-1]$ is 
the first occurrence of the factor ${\bf t}[i..i+n-1]$ in $\bf t$;
\item $\afirst(i,n)$ asserts that the first occurrence of the
factor $x = {\bf t}[i..i+n-1]$ precedes the first occurrence
of $\overline{x}$ in $\bf t$;
\item $\abpat(i,n)$ asserts that the intertwining sequence of
$x = {\bf t}[i..i+n-1]$ and $\overline{x}$ is ${\tt (AB)}^\omega$.
\item $\bapat(i,n)$ asserts that the intertwining sequence of
$x = {\bf t}[i..i+n-1]$ and $\overline{x}$ is ${\tt (BA)}^\omega$.
\end{itemize}

\begin{align*}
\feq(i,j,n) &:= \forall k \ (k<n) \implies {\bf t}[i+k] = {\bf t}[j+k] \\
\feqc(i,j,n) &:= \forall k \ (k<n) \implies {\bf t}[i+k] \not= {\bf t}[j+k] \\
\either(i,j,n) &:= \feq(i,j,n) \, \vee\, \feqc(i,j,n) \\
\consec(i,j,k,n) &:= (j<k) \, \wedge\, \either(i,j,n) \, \wedge\,
\either(i,k,n) \, \wedge\, \forall l \ (j<l \, \wedge \, l<k) \\
 & \quad \implies \neg\either(i,l,n) \\
\ab(i,j,k,n) &:= \feq(i,j,n) \, \wedge\, \feqc(i,k,n) \\
\first(i,j,n) &:=  \feq(i,j,n) \, \wedge\, \forall k \ (k<j) \implies \neg\feq(i,k,n)\\
\afirst(i,n) &:= \forall j,k\ (\first(i,j,n) \, \wedge\, \feqc(i,k,n)) 
	\implies j<k\\
\abpat(i,n) &:= (n>0) \, \wedge\, \afirst(i,n) \, \wedge\, \forall j,k\ \consec(i,j,k,n) 
	\implies (\ab(i,j,k,n) \, \vee \, \ab(i,k,j,n)) \\
\bapat(i,n) &:= (n>0) \, \wedge\,  (\neg\afirst(i,n)) \, \wedge\, \forall j,k\ \consec(i,j,k,n) 
	\implies (\ab(i,j,k,n) \, \vee \, \ab(i,k,j,n)) \\
\end{align*}

The translation into {\tt Walnut} is
\begin{verbatim}
def feq "Ak (k<n) => T[i+k]=T[j+k]":
def feqc "Ak (k<n) => T[i+k]!=T[j+k]":
def either "$feq(i,j,n)|$feqc(i,j,n)":
def consec "j<k & $either(i,j,n) & $either(i,k,n) & Al (j<l & l<k) 
   => ~$either(i,l,n)":
def ab "$feq(i,j,n) & $feqc(i,k,n)":
def first "$feq(i,j,n) & Ak (k<j) => ~$feq(i,k,n)":
def afirst "Aj,k ($first(i,j,n) & $feqc(i,k,n)) => j<k":
def abpat "(n>0) & $afirst(i,n) & Aj,k $consec(i,j,k,n) => 
   ($ab(i,j,k,n)|$ab(i,k,j,n))":
def bapat "(n>0) & (~$afirst(i,n)) & Aj,k $consec(i,j,k,n) => 
   ($ab(i,j,k,n)|$ab(i,k,j,n))":
\end{verbatim}

We now do the same thing for the patterns ${\tt (ABBA)}^\omega$
and ${\tt (BAAB)}^\omega$.   The one complication is that
to assert that the intertwining sequence is
${\tt (ABBA)}^\omega$,
for example, then one must assert that 
\begin{itemize}
\item[(a)] the first two occurrences of either $x$ or $\overline{x}$ 
form the pattern {\tt AB};
\item[(b)] three consecutive occurrences of either $x$ or $\overline{x}$
in $\bf t$ must form the pattern {\tt ABB} or
{\tt BBA} or {\tt BAA} or {\tt AAB}.
\end{itemize}
An easy induction now shows that the overall
pattern can only be ${\tt (ABBA)}^\omega$.

We give only the {\tt Walnut} commands for checking this.
\begin{verbatim}
def firstc "$feqc(i,j,n) & Ak (k<j) => ~$feqc(i,k,n)":
# j is the first occurrence of the complement of t[i..i+n-1]
def abfirst "Aj,k ($first(i,j,n) & $firstc(i,k,n)) => (j<k & Al (j<l & l<k) 
   => ~$either(i,l,n))":
# the first two occurrences of t[i..i+n-1] or its complement are of the form AB
def abb "$feq(i,j,n) & $feqc(i,k,n) & $feqc(i,l,n)":
def bba "$feqc(i,j,n) & $feqc(i,k,n) & $feq(i,l,n)":
def baa "$feqc(i,j,n) & $feq(i,k,n) & $feq(i,l,n)":
def aab "$feq(i,j,n) & $feq(i,k,n) & $feqc(i,l,n)":
def abbapat "(n>0) & $abfirst(i,n) & Aj,k,l ($consec(i,j,k,n) & $consec(j,k,l,n)) 
   => ($abb(i,j,k,l,n) | $bba(i,j,k,l,n) | $baa(i,j,k,l,n) | $aab(i,j,k,l,n))":
def baabpat "(n>0) & (~$abfirst(i,n)) & Aj,k,l ($consec(i,j,k,n) & $consec(j,k,l,n)) 
   => ($baa(i,j,k,l,n) | $aab(i,j,k,l,n) | $abb(i,j,k,l,n) | $bba(i,j,k,l,n))":
\end{verbatim}

Now we are ready to finish the proof of the theorem.  First we
check that
\begin{align*}
I( {\tt 11} ) &= I( {\bf t}[1..2]) = {\tt (AB)}^\omega \\
I( {\tt 00} ) &= I( {\bf t}[5..6]) = {\tt (BA)}^\omega \\
I( {\tt 101} ) &= I( {\bf t}[2..4]) = {\tt (ABBA)}^\omega \\
I( {\tt 010} ) &= I( {\bf t}[3..5]) = {\tt (BAAB)}^\omega ,
\end{align*}
as follows:
\begin{verbatim}
eval alloccur "$abpat(1,2) & $bapat(5,2) & $abbapat(2,3) & $baabpat(3,3)":
\end{verbatim}
and {\tt Walnut} returns {\tt TRUE}.

Next, we check that for all $i$ and all $n \geq 2$, the intertwining
sequence of 
${\bf t}[i..i+n-1]$ is either ${\tt (AB)}^\omega$,
${\tt (BA)}^\omega$,
${\tt (ABBA)}^\omega$,
or ${\tt (BAAB)}^\omega$.
\begin{verbatim}
eval checkeach "Ai,n (n>=2) => ($abpat(i,n)|$bapat(i,n)|$abbapat(i,n)|
   $baabpat(i,n))":
\end{verbatim}
and {\tt Walnut} returns {\tt TRUE}.

This completes the proof.
\end{proof}

For factors of length $n = 2$, the only intertwining patterns that occur
are ${\tt (AB)}^\omega$ and ${\tt (BA)}^\omega$.
However, for each $n \geq 3$, we can prove that each of the
four patterns actually occurs.

\begin{theorem}
For every $n \geq 3$, and each of the four
patterns $p \in \{ {\tt AB, BA, ABBA, BAAB} \}$,
there is a length-$n$ factor $x$ of $\bf t$
whose occurrence pattern is $p^\omega$.
\end{theorem}

\begin{proof}
We use {\tt Walnut} with the command
\begin{verbatim}
eval checklen "An (n>=3) => Ei,j,k,l $abpat(i,n) & $bapat(j,n) & $abbapat(k,n) 
   & $baabpat(l,n)":
\end{verbatim}
and {\tt Walnut} returns {\tt TRUE}.
\end{proof}

\section{Automata}

We now give the four automata for the four cases of intertwining sequence.
Each automaton takes, as input, the base-$2$ expansions of $i$ and
$n$ in parallel, starting with the most significant bit, and accepts
if and only if ${\bf t}[i..i+n-1]$ has the specified intertwining
sequence.

\begin{figure}[H]
\begin{center}
\includegraphics[width=6.5in]{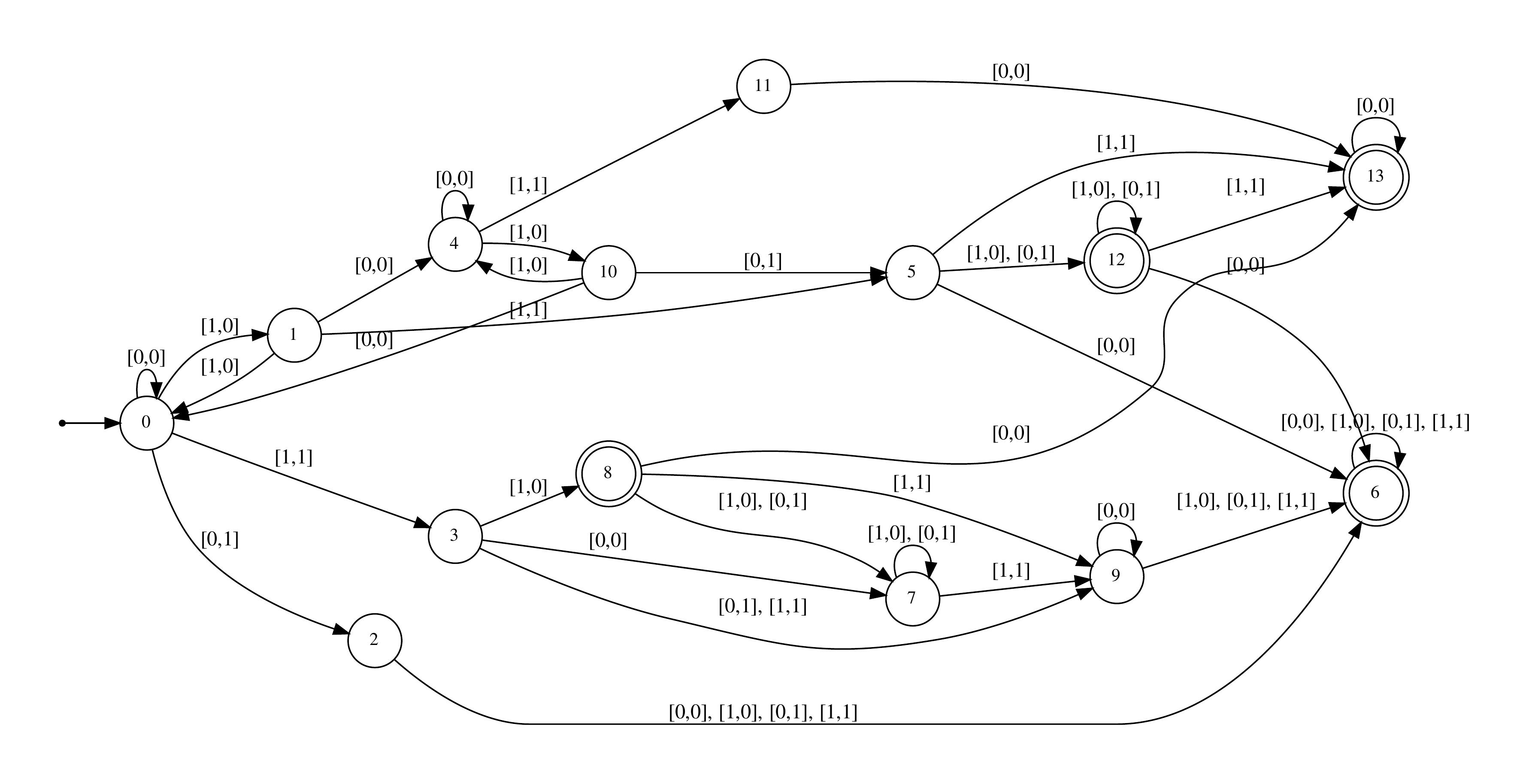}
\end{center}
\caption{Automaton for $(i,n)$ such that $I({\bf t}[i..i+n-1]) = {\tt (AB)}^\omega$.}
\end{figure}

\begin{figure}[H]
\begin{center}
\includegraphics[width=6.5in]{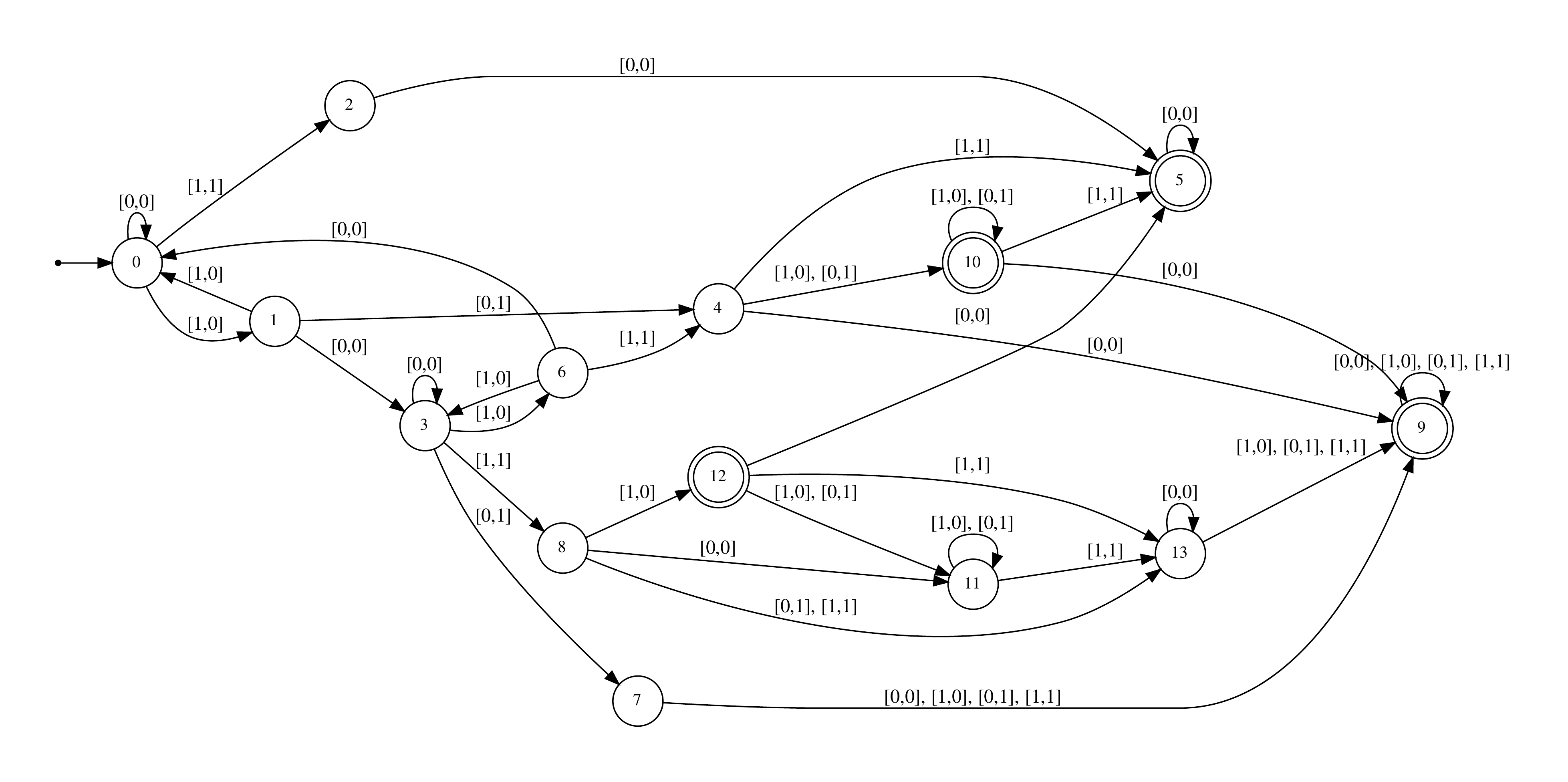}
\end{center}
\caption{Automaton for $(i,n)$ such that $I({\bf t}[i..i+n-1]) = {\tt (BA)}^\omega$.}
\end{figure}

\begin{figure}[H]
\begin{center}
\includegraphics[width=6.5in]{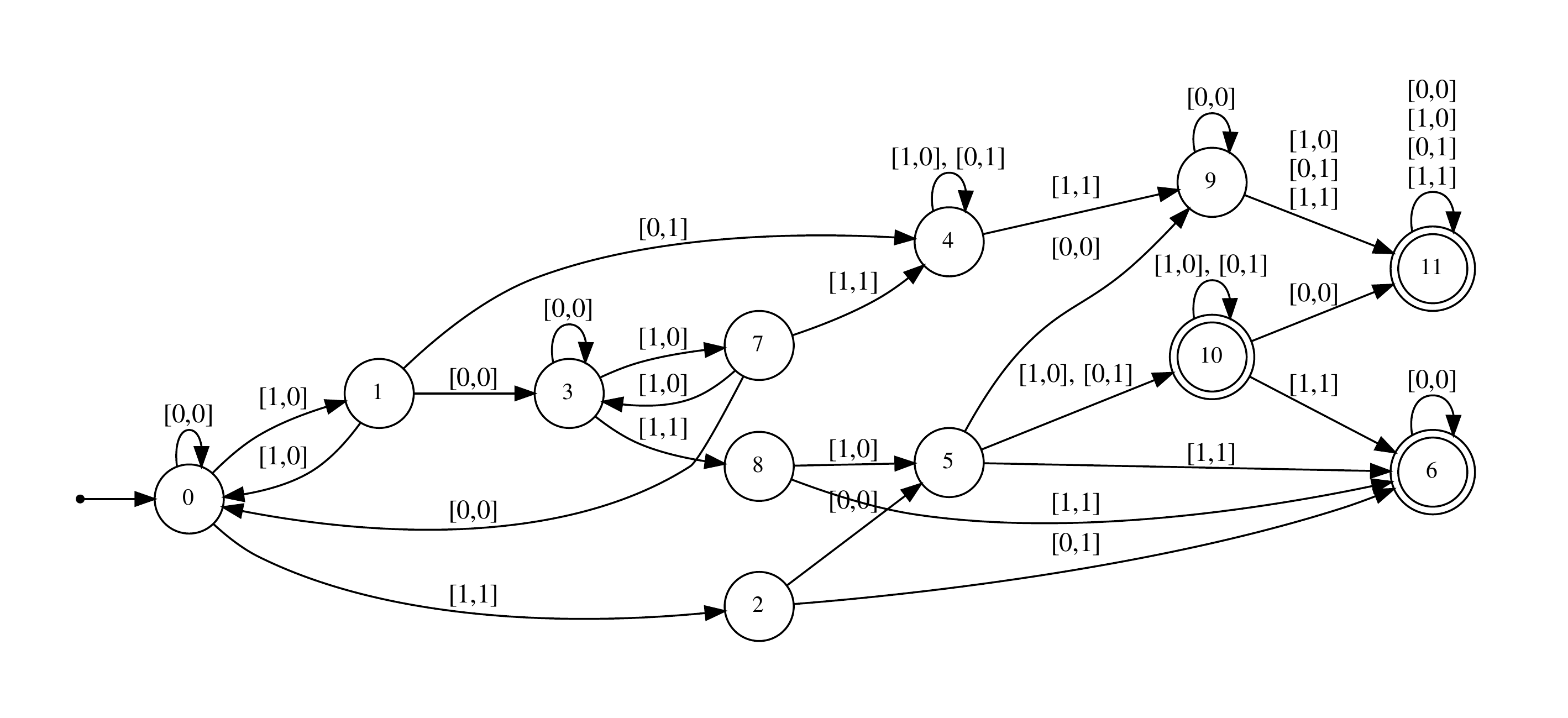}
\end{center}
\caption{Automaton for $(i,n)$ such that $I({\bf t}[i..i+n-1]) = {\tt (ABBA)}^\omega$.}
\end{figure}

\begin{figure}[H]
\begin{center}
\includegraphics[width=6.5in]{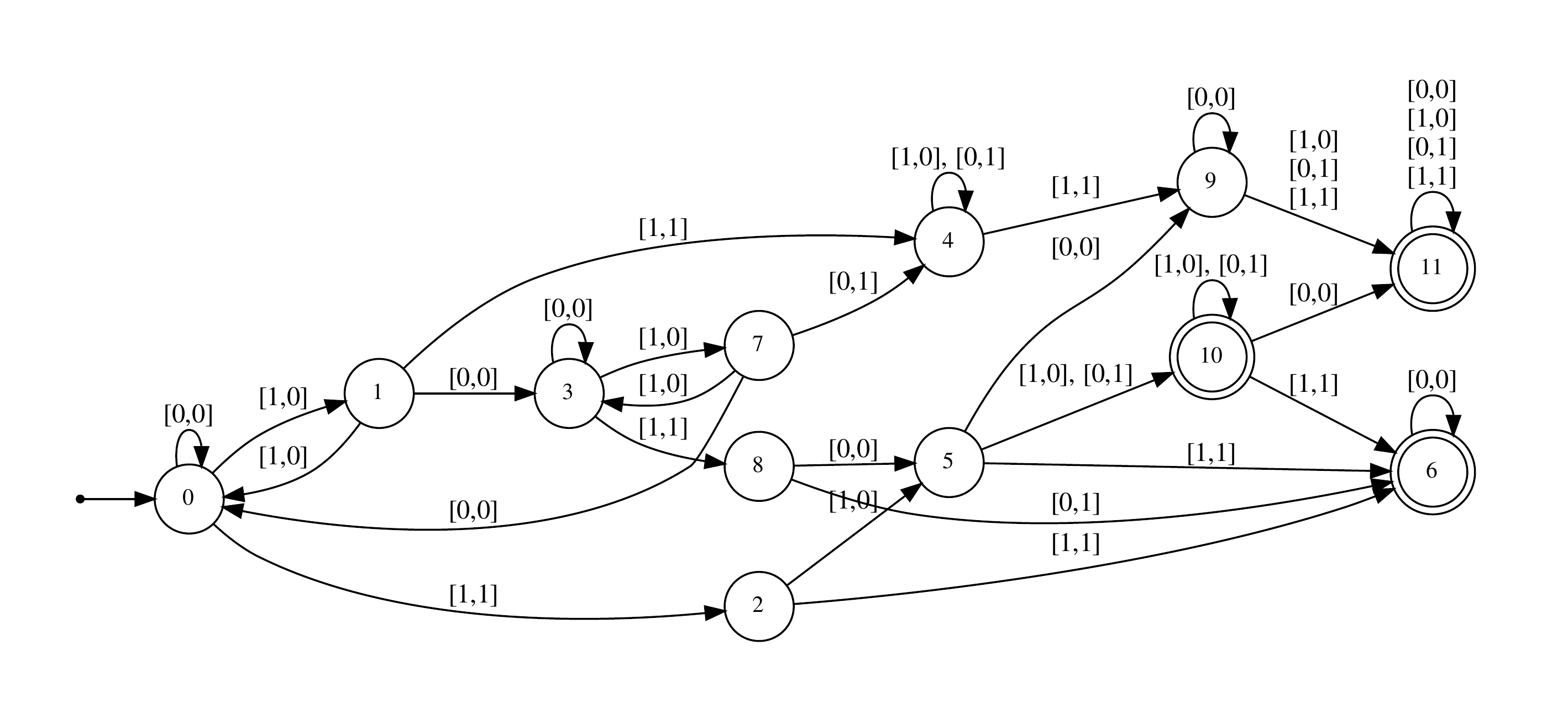}
\end{center}
\caption{Automaton for $(i,n)$ such that $I({\bf t}[i..i+n-1]) = {\tt (BAAB)}^\omega$.}
\end{figure}

\begin{remark}
We can combine these automata, as in \cite{Shallit&Zarifi:2019},
to get a single DFAO that, on input $(i,n)$, computes which of the
six possibilities in Theorem~\ref{main} occurs.
However, the resulting automaton has 30 states and is rather 
complicated in appearance, so we don't give it here.
\end{remark}

\section{Number of factors of each type}

We now determine the number of length-$n$ factors of each of the four types.
It is easy to see that there is a 1-1 correspondence between
length-$n$ factors where the intertwining sequence is
${\tt (AB)}^\omega$ and those where the intertwining sequence
is ${\tt (BA)}^\omega$, and similarly for those with
intertwining sequence ${\tt (ABBA)}^\omega$
and ${\tt (BAAB)}^\omega$.  Thus it suffices to just handle
${\tt (AB)}^\omega$ and ${\tt (ABBA)}^\omega$.

Let $f(n)$ be the number of length-$n$ factors $x$ of $\bf t$
where $I(x) = {\tt (AB)}^\omega$, and let
$g(n)$  be the number of length-$n$ factors $x$ of $\bf t$
where $I(x) = {\tt (ABBA)}^\omega$.    Here is a table of the
first few values of these functions:
\begin{table}[H]
\begin{center}
\begin{tabular}{c|cccccccccccccccc}
$n$ & 1 & 2 & 3 & 4 & 5 & 6 & 7 & 8 & 9 & 10 & 11 & 12 & 13 & 14 & 15\\
\hline
$f(n)$ & 0 & 2 & 2 & 4 & 4 & 6 & 8 & 8 & 8 & 10 & 12 & 14 & 16 & 16 & 16 \\
$g(n)$ & 0 & 0 & 1 & 1 & 2 & 2 & 2 & 3 & 4 & 4 & 4 & 4 & 4 & 5 & 6 
\end{tabular}
\end{center}
\end{table}

It turns out that both of these
are expressible in terms of known sequences.

\begin{theorem}
We have
\begin{align*}
f(n+1) &= 2 \cdot \seqnum{A006165} (n) \quad \text{for $n \geq 1$}; \\
g(n+1) &= \seqnum{A060973}(n) \quad \text{for $n \geq 0$},
\end{align*}
where the sequence numbers refer to sequences in the
{\it On-Line Encyclopedia of Integer Sequences} (OEIS) \cite{Sloane}.
\end{theorem}

\begin{proof}
Let us start with ${\tt (AB)}^\omega$.  Using the {\tt Walnut} commands
\begin{verbatim}
def firstocc "Aj (j<i) => ~$feq(i,j,n)":
eval mab n "$firstocc(i,n+1) & $abpat(i,n+1)":
\end{verbatim}
we can construct the linear representation for $f(n+1)$.  It is
\begin{align*}
v_1 &= \left[ \begin{smallarray}{cccc}
1&0&0&0
	\end{smallarray} \right] 
&
\gamma_1({\tt 0}) & = \left[ \begin{smallarray}{cccc}
1&0&0&0\\
0&2&0&0\\
0&0&0&2\\
0&0&0&1
	\end{smallarray} \right]
&
\gamma_1({\tt 1}) &= \left[\begin{smallarray}{cccc} 
0&1&1&0\\
0&2&0&0\\
0&2&0&0\\
0&1&0&1
	\end{smallarray}\right]
&
w_1 &= \left[ \begin{smallarray}{c}
0\\
1\\
1\\
0
\end{smallarray} 
\right] .
\end{align*}
On the other hand, from the known relations for $\seqnum{A006165}(n)$,
namely
\begin{align*}
\seqnum{A006165}(2n) &= 2\seqnum{A006165}(n) - [n=0] - [n=1] \\
\seqnum{A006165}(2n+1) &= \seqnum{A006165}(n+1) + \seqnum{A006165}(n) - [n=0]
\end{align*}
from which we can compute its linear representation:
\begin{align*}
v_2 &= \left[ \begin{smallarray}{cccc}
1&1&1&0
	\end{smallarray} \right] 
&
\gamma_2({\tt 0}) & = \left[ \begin{smallarray}{cccc}
2&1&0&0\\
0&1&0&0\\
-1&-1&1&0\\
-1&0&0&0
	\end{smallarray} \right]
&
\gamma_2({\tt 1}) &= \left[\begin{smallarray}{cccc} 
1&0&0&0\\
1&2&0&0\\
-1&-1&0&1\\
0&0&0&0
	\end{smallarray}\right]
&
w_2 &= \left[ \begin{smallarray}{c}
1\\
0\\
0\\
0
\end{smallarray} 
\right] .
\end{align*}
Here we are using the Iverson bracket, where (for example)
the expression $[n=0]$ evaluates to $1$ if $n=0$ and $0$ otherwise.

From these two linear representations, we can easily compute
the linear representation for 
$f(n+1) - 2 \cdot \seqnum{A006165} (n)$ and then minimize
it.  When we do so, we get a linear representation of rank $1$
that evaluates to the function $-2[n=0]$, so indeed
$f(n+1) = 2 \cdot \seqnum{A006165} (n)$ for all $n \geq 1$.

We can do the same thing for $g(n+1)$, using the {\tt Walnut}
command:
\begin{verbatim}
eval mabba n "$firstocc(i,n+1) & $abbapat(i,n+1)":
\end{verbatim}

\begin{align*}
v_3 &= \left[ \begin{smallarray}{cccccc}
1&1&0&0&0&0
	\end{smallarray} \right] 
&
\gamma_3({\tt 0}) & = \left[ \begin{smallarray}{cccccc}
1&1&0&0&0&0\\
0&0&0&0&0&0\\
0&0&0&0&1&0\\
0&0&0&1&0&0\\
0&0&0&0&1&1\\
0&0&0&0&0&2
	\end{smallarray} \right]
&
\gamma_3({\tt 1}) &= \left[\begin{smallarray}{cccccc} 
0&0&1&0&0&0\\
0&0&0&1&0&0\\
0&0&0&0&0&0\\
0&0&0&1&0&1\\
0&0&0&0&1&0\\
0&0&0&0&0&2
	\end{smallarray}\right]
&
w_3 &= \left[ \begin{smallarray}{c}
0\\
0\\
0\\
0\\
1\\
1
\end{smallarray} 
\right] .
\end{align*}

On the other hand, from the known relations for $\seqnum{A060973}(n)$,
namely
\begin{align*}
\seqnum{A060973}(2n) &= 2\seqnum{A060973}(n) + [n=1] \\
\seqnum{A060973}(2n+1) &= \seqnum{A060973}(n+1) + \seqnum{A060973}(n) 
\end{align*}
from which we can compute its linear representation:
\begin{align*}
v_4 &= \left[ \begin{smallarray}{cccc}
0&0&1&0
	\end{smallarray} \right] 
&
\gamma_4({\tt 0}) & = \left[ \begin{smallarray}{cccc}
2&1&0&0\\
0&1&0&0\\
0&0&1&0\\
1&0&0&0
	\end{smallarray} \right]
&
\gamma_4({\tt 1}) &= \left[\begin{smallarray}{cccc} 
1&0&0&0\\
1&2&0&0\\
0&1&0&1\\
0&0&0&0
	\end{smallarray}\right]
&
w_4 &= \left[ \begin{smallarray}{c}
1\\
0\\
0\\
0
\end{smallarray} 
\right] .
\end{align*}

Once again we can compute the linear representation for
$g(n+1) - \seqnum{A060973} (n)$ and minimize it.
When we do so, we get a linear representation of rank $0$,
computing the constant function $0$.
\end{proof}

Finally, using the known expressions for the two sequences
\seqnum{A006165} and \seqnum{A060973}, we arrive at the following
result:
\begin{corollary}
For $n \geq 2$ we have
$$ f(n) = \begin{cases}
	2^k,  & \text{if $3 \cdot 2^{k-2} < n \leq 2^k + 1$}; \\
	2n-2^k-2, & \text{if $2^k + 1 < n \leq 3 \cdot 2^{k-1}$}.
	\end{cases}
$$
For $n \geq 3$ we have
$$ g(n) = \begin{cases}
	2^{k-1}, & \text{if $2^k + 1 < n \leq 3 \cdot 2^{k-1} + 1$}; \\
	n-2^{k-1}-1, & \text{if $3 \cdot 2^{k-1} + 1 < n \leq 2^k + 1$}.
	\end{cases} $$
\end{corollary}

\end{document}